\documentclass[11pt]{article}
\usepackage{amsfonts}
\usepackage{graphics,amssymb,amsthm,amsmath,epsf}
\usepackage{graphicx}
\usepackage{subfigure}

\usepackage[round]{natbib}
\usepackage{xcolor}

\newtheorem{algorithm}{Algorithm}[section]
\newtheorem{lemma}{Lemma}[section]
\numberwithin{equation}{section}

\allowdisplaybreaks

\begin{document}

\title{ Modified Signed Log-Likelihood Ratio Test for \\
 Comparing the Correlation Coefficients  of Two \\
 Independent Bivariate Normal Distributions}
\author{ M. R. Kazemi$^{1}$  , A. A. Jafari$^{2,}$\thanks{%
Corresponding: aajafari@yazd.ac.ir}\\
{\small $^1$Department of Statistics, Fasa University, Fasa, Iran}\\
{\small $^2$Department of Statistics, Yazd University, Yazd, Iran}\\
}
\date{}
\maketitle

\begin{abstract}
In this paper, we use the method of modified signed log-likelihood ratio test for the problem of testing the equality of correlation
coefficients  in two independent bivariate normal distributions. We compare this method  with two other 
approaches,  Fisher's Z-transform and generalized test variable, using  a Monte Carlo simulation. It indicates that  the proposed method is better
than the other approaches, in terms of the actual sizes  and powers especially when the sample sizes are unequal. We illustrate  performance of the proposed approach, using a real data set.

\end{abstract}

\section{Introduction}

The linear association between two  normal variables is usually measured by correlation coefficient. Statistical inferences for this parameter
are
divided to a single bivariate sample and several bivariate samples problems. In  the case of a single sample,
\cite{fisher-15}
 for the first time and then
\cite{hotelling-53}
provided the exact density function of  product moment correlation coefficient.   In testing and constructing the confidence
interval for the correlation coefficient, \cite{fisher-21} 
introduced the  well-known Fisher Z-transform, \cite{su-wo-07}
proposed a likelihood-based higher-order asymptotic method, and \cite{kr-xi-07}
proposed a generalized pivotal approach. \cite{ka-ja-15} compared some confidence intervals for the correlation coefficient.

The problem of equality of two correlations arises practically,   for example in comparing the correlations between the laterality of blood flow in each brain region and verbal
memory score across gender \cite[see][]{bi-br-gu-04}.
For inference about this problem, \cite{zar-99}
used the Fisher Z-transform to test that whether all samples came from populations having common correlation coefficient, and \cite{ol-fi-95}
 obtained an asymptotic distribution of the difference between two sample correlation coefficients. \cite{kr-xi-07} 
proposed a generalized test variable and studied the performances of this test, Fisher Z-transform test and Olkin and Finn's method. They concluded
that Olkin and Finn's method is satisfactory for large sample sizes, and Fisher Z-transform test is conservative (i.e. its actual size is very
smaller than the nominal level) when the samples are small. In addition, the actual size of generalized test variable is close to the nominal level for
moderate samples.

The aim of this paper is to develop a modified signed log-likelihood ratio (MSLR) method for testing the equality of two  correlation coefficients
 in two independent bivariate normal distributions. We
used the test statistic proposed by
\cite{di-ma-st-01}
which has a simple form and  then applied the traditional signed log-likelihood ratio (SLR) test in its form. We propose a parametric bootstrap method to
approximate the distribution of SLR statistic and then use it to compute the MSLR statistic. Our simulation results show that MSLR test always are
satisfactory regardless of the sample sizes and values of the common correlation coefficient.

This paper is organized as follows: Some  preliminaries are given  in Section \ref{sec.pr}. The MSLR is explained for testing the
equality of two correlation coefficients  in Section \ref{sec.meth}. 
In Section \ref{sec.sim}, a simulation study is performed to evaluate and compare the actual sizes and powers of  MSLR, Fisher's Z-transform and generalized variable approaches. Also, the
approaches are illustrated using a real example. 

\section{Preliminaries}
\label{sec.pr}
Let $(X_{ij},Y_{ij})$,  $i=1,2$ and $j=1,2,\dots ,n_i$ be a random sample from the bivariate normal distribution with mean vector ${{\boldsymbol{%
\mu }}}_i=({\mu }_{1i},{\mu }_{2i})^{\prime }$ and variance covariance matrix
\begin{equation*}
{\Sigma }_i=\left[
\begin{array}{cc}
{\sigma }^2_{1i} & {\rho }_i{\sigma }_{1i}{\sigma }_{2i} \\
{\rho }_i{\sigma }_{1i}{\sigma }_{2i} & {\sigma }^2_{2i}
\end{array}
\right],\ \ \ \ i=1,2.
\end{equation*}

Our goal is to test the hypothesis
\begin{equation}\label{eq.H0}
H_{0}:{\rho }_{1}={\rho }_{2}=\rho \ \ \  {\rm vs.} \ \ \  H_{1}:\rho _{1}\neq \rho _{2},
\end{equation}
where $\rho $ is the common correlation coefficient. We use the method of SLR for this problem.  To apply this method, we need to find the full and constrained maximum likelihood estimators (MLE) of the unknown model parameters.  Considering
${\boldsymbol\theta}=({\boldsymbol\theta} _1,{\boldsymbol\theta }_2 )$, where
${\boldsymbol \theta }_i=(\mu_{1i},\mu_{2i},\sigma_{1i},\sigma_{2i},\rho_i)$,
it can be shown that the log-likelihood function can be written as
\begin{equation}\label{eq.lik}
\ell ( {\boldsymbol\theta }) =c+\ell_{1}( {{\boldsymbol{\theta }}}_{1}) +{\ell }_{2}( {{\boldsymbol{\theta }}}_{2}),
\end{equation}
where
\begin{eqnarray*}
{\ell }_{i}\left( {{\boldsymbol{\theta }}}_{i}\right)  &=&-n_{i}\log(
\sigma _{1i})-n_{i}\log(\sigma _{2i})-\frac{n_{i}}{2}\log
\left( 1-{\rho }_{i}^{2}\right) -\frac{n_{i}{\mu }_{1i}^{2}}{2\left( 1-{\rho
}_{i}^{2}\right) {\sigma }_{1i}^{2}}-\frac{n_{i}{\mu }_{2i}^{2}}{2\left( 1-{%
\rho }_{i}^{2}\right) {\sigma }_{2i}^{2}} \\
&&+\frac{n_{i}{\rho }_{i}{\mu }_{1i}{\mu }_{2i}}{2\left( 1-{\rho }%
_{i}^{2}\right) {\sigma }_{1i}{\sigma }_{2i}}-\frac{{\rho }%
_{i}\sum_{k=1}^{n_{i}}{x_{ij}^{2}}}{2\left( 1-{\rho }_{i}^{2}\right) {\sigma
}_{1i}^{2}}-\frac{{\rho }_{i}\sum_{j=1}^{n_{i}}{y_{ij}^{2}}}{2\left( 1-{\rho
}_{i}^{2}\right) {\sigma }_{2i}^{2}}+\frac{({\mu }_{1i}{\sigma }_{2i}-{\mu }%
_{2i}{\sigma }_{1i})\sum_{k=1}^{n_{i}}{x_{ij}}}{\left( 1-{\rho }%
_{i}^{2}\right) {\sigma }_{1i}^{2}{\sigma }_{2i}} \\
&&+\frac{({\mu }_{2i}{\sigma }_{1i}-{\mu }_{1i}{\sigma }_{2i})%
\sum_{k=1}^{n_{i}}{y_{ij}}}{\left( 1-{\rho }_{i}^{2}\right) {\sigma }%
_{2i}^{2}{\sigma }_{1i}}+\frac{{\rho }_{i}{\sigma }_{1i}\sum_{j=1}^{n_{i}}{x_{ij}y_{ij}}}{\left( 1-{\rho }%
_{i}^{2}\right) {\sigma }_{1i}{\sigma }_{2i}}.
\end{eqnarray*}

It is known that, under the full model (without any constraint), the MLE's of parameters ${\boldsymbol\theta}_i$
are ${\hat{{%
\boldsymbol{\theta }}}}_i{\boldsymbol{=}}{({\bar{X}}_i,{\bar{Y}}%
_i,S_{1i},S_{2i},R}_i)$, $i=1,2$, where
\begin{eqnarray*}
{\bar{X}}_i=\frac{1}{n_i}\sum^{n_i}_{j=1}{X_{ij}},\ \  {\bar{Y}}%
_i=\frac{1}{n_i}\sum^{n_i}_{j=1}{Y_{ij}}, \ \
R_i=\frac{S_{12(i)}}{\sqrt{S^2_{1i}S^2_{2i}}}, \\
 S^2_{1i}=\frac{1}{n_i}%
\sum^{n_i}_{j=1}{{\left(X_{ij}-{\bar{X}}_i\right)}^2},\ \  S^2_{2i}=%
\frac{1}{n_i}\sum^{n_i}_{j=1}{{\left(Y_{ij}-{\bar{Y}}_i\right)}^2},
\end{eqnarray*}
and $S_{12(i)}=\sum^{n_i}_{j=1}{\left(X_{ij}-{\bar{X}}%
_i\right)\left(Y_{ij}-{\bar{Y}}_i\right)}.$

For the constrained model i.e. under the hypothesis in \eqref{eq.H0},
\cite{pearson-33}
showed that the MLE of the common correlation coefficient, $\tilde{\rho }
$, is obtained by solving the following equation:
\begin{equation}\label{eq.rhot}
\frac{n_{1}\left( \ r_{1}-\tilde{\rho }\right) }{\left( 1-\tilde{%
\rho }r_{1}\right) }+\frac{n_{2}\left( \ r_{2}-\tilde{\rho }\right) }{%
\left( 1-\tilde{\rho }r_{2}\right) }=0,
\end{equation}
where $r_i$ is the observed value of $R_i$.
 \citep[For more details, refer to][]{pearson-33,do-ro-80}. 
Also, the constrained MLE's of parameters
${\mu }_{1i}$, ${\mu }_{2i}$, ${\sigma }^2_{1i}$ and ${\sigma }^2_{2i}$ are
\begin{equation*}
\tilde{\mu}_{1i}={\bar{x}}_i,\ \ \ {\tilde{\mu }}_{2i}={
\bar{y}}_i,\ \ {\tilde{\sigma }}^2_{1i}=
\frac{s^2_{1i}\left(1-\tilde{\rho } r_i\right)}{1-{\tilde{\rho }^2}},
\ \  {\tilde{\sigma }}
^2_{2i}=\frac{s^2_{2i}\left(1-\tilde{\rho }r_i\right)}{1-{\tilde{\rho }}^2},
\end{equation*}
where ${\bar{x}}_i,{\bar{y}}_i,s^2_{1i}$ and $s^2_{2i}$ are the observed value of ${\bar{X}}_i,{\bar{Y}}_i,S^2_{1i}$ and $S^2_{2i}$.
In this case, the MLE of parameter ${\boldsymbol\theta}_i$ is
$\tilde{\boldsymbol \theta }_i =(\tilde{\mu }_{1i}, \tilde{\mu}_{2i},\tilde{\sigma}_{1i},\tilde{\sigma}_{2i},\tilde\rho)$.

\cite{do-ro-80} 
defined
\begin{equation}\label{eq.RF}
R_{F}=\frac{\exp ( 2\bar{Z})-1}{\exp ( 2\bar{Z})+1}=\tanh(\bar{Z}) ,
\end{equation}
where $\bar{Z}=\frac{(n_1-3)Z_1+(n_2-3)Z_2}{n_1+n_2-6}$ and
$Z_i=\frac{1}{2}\log(\frac{1+R_i}{1-R_i})=\tanh^{-1}(R_i)$,
$i=1,2$.
They  showed that the estimators $\tilde{\rho }$ and $R_F$ are close when the samples sizes are equal, i.e. $n_1=n_2$.
 Simulation studies (not reported here) show that the results for MSLR method  based on the estimators $\tilde{\rho }$ and $R_F$ are close to each other.  But the estimator $R_{F}$ decrease the execution time. Therefore, $R_F$ can be
used instead of $\tilde{\rho }$ to estimate the common correlation coefficient $\rho $.

The following lemma helps us to generate the sample correlation coefficient from a random sample of a bivariate normal distribution. { It is notable that this formula is different from formula (16) of
\cite{kr-xi-07}
and also it cannot be used as the generalized pivotal quantity.}

\begin{lemma} Let $R_i$ be the sample correlation coefficient from a bivariate normal distribution with mean vector
${\boldsymbol \mu }_i$ and variance-covariance matrix $\Sigma_i$. Then
\begin{equation}\label{eq.Ri}
R_{i}\overset{d}{=}\frac{{\rho }_{i}^{\ast }V_{i}+N_{i}\ }{\sqrt{{\left( {\rho }%
_{i}^{\ast }V_{i}+N_{i}\right) }^{2}+{W_{i}}^{2}}},
\end{equation}
where ${\rho }^*_i=\frac{\rho_i}{\sqrt{1-\rho^2_i}}$ , and $V^2_i$, $W^2_i$, and $N_i$ are independent random variables with
$\chi^2_{(n-1)}$, $\chi^2_{(n-2)}$ and $N(0,1)$ distributions, respectively.
\end{lemma}

\begin{proof}
Let $S_{i}=\left[
\begin{array}{cc}
S_{1i}^{2} & S_{12(i)} \\
S_{12(i)} & S_{2i}^{2}%
\end{array}%
\right] $.
It is well-known that  $A_{i}=n_{i}S_{i}\sim W(n_{i}-1,{\Sigma}_{i})$, i.e. it has a Wishart distribution with $n_{i}-1$ degrees of
freedom and parameter ${\Sigma }_{i}$. Since ${\Sigma }_{i}$ is a positive definite matrix, there is a unique lower triangular matrix, $L_{i}$, such
that $L_{i}L_{i}^{\prime }={\Sigma }_{i}$ (Cholesky decomposition)  and it is easily verified that
\[
P_{i}=L_{i}^{-1}A_{i}L_{i}^{\prime -1}\sim W\left( n_{i}-1,I\right),
\]
where $I$ is the identity matrix. Put $P_{i}=C_{i}C_{i}^{\prime }$.
From Theorem 3.2.14 of
\cite{muirhead-82},
the elements of matrix $C_{i}=\left[
\begin{array}{cc}
V_{i} & 0 \\
N_{i} & W_{i}%
\end{array}%
\right] $ are independent and distributed as
\[
V_{i}^{2}\sim \chi _{\left( n_{i}-1\right) }^{2},\ W_{i}^{2}\sim \chi
_{\left( n_{i}-1\right) }^{2}\ \ \text{and}\ \ \ \ N_{i}\sim N\left(
0,1\right) .
\]
It can be shown that matrix $L_{i}$ has the following form:
\[L_{i}=\left[ \begin{array}{cc}
{\sigma }_{1i} & 0 \\
{\rho}_i {\sigma }_{2i} & {\sigma }_{2i}\sqrt{1-{\rho }_i^2} \end{array}
\right].\]
Therefore, we have
\begin{eqnarray*}
A_{i} &=&\left( A_{kl}^{\left( i\right) }\right) \overset{d}{=}L_{i}C_{i}C_{i}^{\prime
}L_{i}^{\prime } \\
&=&
\left[
\begin{array}{cc}
\sigma _{1i}^{2}V_{i}^{2} & \sigma _{1i}\sigma _{2i}\sqrt{1-\rho _{i}^{2}}%
\left( \tilde{\rho}_{i}V_{i}^{2}+N_{i}V_{i}\right)  \\
\sigma _{1i}\sigma _{2i}\sqrt{1-\rho _{i}^{2}}\left( \tilde{\rho}%
_{i}V_{i}^{2}+N_{i}V_{i}\right)  & \sigma _{2i}^{2}\left( 1-\rho
_{i}^{2}\right) \left[ \left( \tilde{\rho}_{i}V_{i}+N_{i}\right)
^{2}+W_{i}^{2}\right]
\end{array}%
\right].
\end{eqnarray*}
Since $R_{i}=\frac{A_{12}^{\left( i\right) }}{\sqrt{A_{11}^{\left( i\right)
}A_{22}^{\left( i\right) }}}$, the proof is completed.

\end{proof}

\section{Testing the equality of two correlation coefficients}
\label{sec.meth}
In this section, we consider the problem of testing the equality of two independent correlation coefficients. At first, we propose the method of
MSLR and give an algorithm that can be used for this problem. Then, we review two other existing approaches.

\subsection{Modified Signed log-likelihood ratio test}

It is easily verified that the SLR test statistic to test the hypothesis in \eqref{eq.H0} has the following form
\begin{eqnarray}\label{eq.SLR}
{\rm SLR} &=&\sqrt{2}{\rm sign}(r_{1}-r_{2}) \sqrt{2(\ell (\hat{{\boldsymbol{\theta }}})-\ell (\tilde{{\boldsymbol{\theta }}})) } \nonumber   \\
&=&{\rm sign}( r_{1}-r_{2})\left( \sum_{i=1}^{2}{n_{i}\log( \frac{( 1-\tilde{\rho }r_{i})^{2}}{
(1-r_{i}^{2}) ( 1-\tilde{\rho}^{2}) }) }\right)^{\frac{1}{2}},
\end{eqnarray}
where
$\hat{\boldsymbol\theta }=(\hat {\boldsymbol \theta}_1, \hat{\boldsymbol\theta }_2)$
and
$\tilde{\boldsymbol\theta}=(\tilde{\boldsymbol\theta }_1,\tilde{\boldsymbol\theta }_2)$,
and
$\mathrm{sign}(x)=1$, if $x>0$ and $\mathrm{sign}(x)=-1$, if $x<0$.

It is well known that SLR is asymptotically distributed
as a standard normal distribution \citep{co-hi-74},
and a   p-value for testing
the hypothesis in \eqref{eq.H0} is
\begin{equation}\label{eq.pSLR}
p=2( 1-\Phi ( | \mathrm{SL}R_0|) ),
\end{equation}
where $\mathrm{SLR}_0$ is the observed value of the statistic SLR and $\Phi(t)$ is the standard normal distribution function.

 If we use the estimator $R_F$ instead of $\tilde\rho$, the SLR statistic in \eqref{eq.SLR} is
rewritten as
\begin{equation}\label{eq.SLRF}
\mathrm{SLR}=\mathrm{sign}( r_{1}-r_{2})
\left( \sum_{i=1}^{k}
n_{i}\log ( \frac{( 1-R_{F} r_{i})^{2}}{(
1-r_{i}^{2}) ( 1-R_{F}^{2})}) \right)^{\frac{1}{
2}}.
\end{equation}

\cite{pi-pe-92} 
showed the SLR test is not very accurate, and some modifications are needed to increase the accuracy of the SLR.  There exist various ways to improve the accuracy
of this approximation by adjusting the SLR statistic. For the various ways to improve the accuracy of SLR method, refer to the works of
\cite{barndorff-86,barndorff-91},
\cite{Skovgaard-01},  
and
\cite{di-ma-st-01}. 
We used the method proposed by \cite{di-ma-st-01},
which has the following form
\begin{equation}\label{eq.MSLR}
\mathrm{MSLR}=\frac{\mathrm{SLR}-m\left( \mathrm{SLR}\right) }{\sqrt{v\left(\mathrm{SLR}\right)}},
\end{equation}
where $m(\mathrm{SLR})$ and $v(\mathrm{SLR})$ are the mean and variance of the SLR statistic evaluated at the
constrained MLE's of the model parameters and is asymptotically distributed as a standard normal distribution.

\cite{kr-le-14} 
used the parametric bootstrap approach to approximate the mean and variance of the MSLR test statistic for the problem of testing the equality of normal
coefficients of variation. We use this approach for the problem of testing the equality of two normal correlation coefficients. In Section \ref{sec.sim}, using Monte Carlo simulation,
we will show  that
this new method is
more accurate than the other competing methods. This approach is given in the
following algorithm:

\begin{algorithm}
Given $r_1$ and $r_2$,

\noindent 1. Compute $r_F$, the observed value of estimator $R_F$ in \eqref{eq.RF}.

\noindent 2. Generate $V^2_i\sim {\chi }^2_{(n_i-1)}$, $W^2_i\sim {\chi }^2_{(n_i-2)}$, $N_i\sim N\left(0,1\right)$, $i=1,2$.

\noindent  3. Compute $r^*_F=r_F/\sqrt{1-r^2_F}$.

\noindent  4. Compute $r^*_i$ by substituting $r_F$ instead of $\rho $ in \eqref{eq.Ri} as
\begin{equation*}
r^*_i=\frac{r^*_FV_i+N_i\ }{\sqrt{{\left(r^*_FV_i+N_i\right)}^2+W^2_i}}.
\end{equation*}

\noindent 5. Compute the test statistic SLR in \eqref{eq.SLRF}.

 \noindent 6. Repeat steps 3-5 for a large number of times (say M = 10,000).

\noindent 7. Compute the sample mean and sample variance of SLR and compute the MSLR
in \eqref{eq.MSLR}.

\noindent 8.  Determine the p-value for testing $H_0:{\rho }_1={\rho }_2$ vs $H_1:{\rho }_1\ne {\rho}_2$ as
\begin{equation}\label{eq.3-9}
\mathrm{p-value}=2\left( 1-\Phi ( | \mathrm{MSLR}| ) \right) .
\end{equation}
\end{algorithm}

\subsection{Fisher's Z-transform}

 It is well-known that
\begin{equation*}
Z_i=\frac{1}{2}{\log (\frac{1+R_i}{1-R_i})}=\tanh^{-1} (R_i),
\end{equation*}
has asymptotic normal distribution with mean ${\tanh^{-1} (\rho _i)}$ and variance $\frac{1}{n_i-3}$. Therefore, a test statistic for testing
$H_0:\rho _1={\rho }_2,$ vs $H_1:{\rho }_1\ne {\rho }_2$ can be given by extending the one-sample Fisher's Z-transformation to the two-sample case.
Consider the following test statistic
\begin{equation}\label{eq.FZ}
FZ=\frac{ Z_{1}-Z_{2}  }{\sqrt{\frac{1%
}{n_{1}-3}-\frac{1}{n_{2}-3}}}.
\end{equation}
Then, $FZ$ has asymptotic standard normal distribution, and the null hypothesis  is rejected if
$\left\vert FZ\right\vert >z_{\alpha/2}$. For more details, refer to \cite{zar-99} and \cite{kr-xi-07}. 

\subsection{Generalized test variable}

\cite{kr-xi-07}
proposed a generalized pivotal variable for $\rho_i$ as
\begin{equation}\label{eq.Gi}
G_{\rho _{i}}=\frac{r_{i}^{\ast }W_{i}-Z_{i} }{\sqrt{{( r_{i}^{\ast}W_{i}-Z_{i}) }^{2}+V_{i}^{2}}},\ \ \ \ i=1,2,
\end{equation}
where $r^*_i=\frac{r_i}{\sqrt{1-r^2_i}}$, and $V_i$, $W_i$, and $Z_i$ are independent random variables with
\begin{equation*}
V^2_i\sim {\chi }^2_{(n_i-1)},\ \ \  W^2_i\sim {\chi }^2_{(n_i-2)},\ \ \  Z_i\sim N(0,1).
\end{equation*}
Therefore, a generalized pivotal variable for ${\rho }_1-{\rho }_2$ is given as
\begin{equation*}
G_{\rho }=G_{{\rho }_1}-G_{{\rho }_2}.
\end{equation*}
So, the generalized p-value for testing $H_0:{\rho }_1={\rho }_2$ vs
$H_1:\rho_1\ne \rho_2$ is given by
\begin{equation}\label{eq.GPV}
p=2 \min\left\{ P( G_\rho<0) ,P( G_\rho>0) \right\}.
\end{equation}

\section{Numerical study}
\label{sec.sim}
\subsection{Simulation study}

We performed a simulation study with 10,000 replications to evaluate and compare the actual sizes of three approaches: the modified signed likelihood
ratio  test (MSLR), Fisher's Z-transform (FZ)  test, and generalized test variable (GV). We generate random samples of size $n_1$ and $n_2$
from two independent bivariate normal distributions for different values of common correlation $\rho =0.0,0.1,0.2,\dots,0.9$.
We obtained the sample correlation coefficient  and then the \textit{p}-values  of the MSLR, FZ and GV to test the hypothesis
$H_{0}:\rho_{1}=\rho_2$ vs  $H_{1}:\rho_{1}\neq{\rho }_{2}$. Here, we consider the nominal level $\alpha =0.05$. The results are given in Table \ref{tab.1}.

We can conclude that\\
i. the actual size of MSLR test is satisfactory for all different values of common correlation coefficient and sample sizes.\\
ii. the actual size of FZ test is smaller than the nominal level when the sample 
sizes are small,\\
iii. the actual size of GV test is very smaller than the nominal level when $n_1$ is small and $n_2$ is large.

Since, the MSLR test is the only test that controls the correct frequency of rejected hypotheses in all cases, we recommend the MSLR for practical applications.

We also performed a simulation study to compare the powers of the considered approaches. The results are given in Tables \ref{tab.2.1} and \ref{tab.2.2}.
It can be concluded that the powers of the
 three tests MSLR, GV and Fisher Z-transform are close  when the sample sizes are equal. But the power of MSLR is larger than powers of GV and Fisher Z-transform when the sample sizes are  unequal.

\begin{table}[ht]
\begin{center}
\caption{ The actual sizes of the tests at nominal level $\alpha =0.05$. }\label{tab.1}
\begin{tabular}{|c|l|cccccccccc|}
\hline
&  & \multicolumn{10}{|c|}{$\rho $} \\ \hline
$n_1,n_2$ & Method & 0.0 & 0.1 & 0.2 & 0.3 & 0.4 & 0.5 & 0.6 & 0.7 & 0.8 & 0.9 \\ \hline

5,5 & MSLR & 0.052 & 0.053 & 0.052 & 0.051 & 0.053 & 0.051 & 0.054 & 0.049 & 0.054 & 0.053 \\
&  FZ & 0.046 & 0.043 & 0.045 & 0.043 & 0.041 & 0.047 & 0.045 & 0.040 & 0.038 & 0.041 \\
& GV & 0.052 & 0.051 & 0.052 & 0.051 & 0.050 & 0.051 & 0.051 & 0.053 & 0.051& 0.051 \\ \hline

5,10 & MSLR & 0.050 & 0.050 & 0.051 & 0.048 & 0.049 & 0.049 & 0.050 & 0.052 & 0.049 & 0.055 \\
&  FZ & 0.049 & 0.047 & 0.045 & 0.045 & 0.043 & 0.045 & 0.046 & 0.046 & 0.045 & 0.044 \\
& GV & 0.049 & 0.048 & 0.050 & 0.051 & 0.051 & 0.049 & 0.049 & 0.050 & 0.051 & 0.048 \\ \hline

10,10 &  MSLR & 0.051 & 0.055 & 0.049 & 0.050 & 0.048 & 0.049 & 0.048 & 0.051 & 0.049 & 0.051 \\
&  FZ & 0.048 & 0.050 & 0.050 & 0.048 & 0.049 & 0.049 & 0.049 & 0.052 & 0.049 & 0.045 \\
& GV & 0.053 & 0.054 & 0.051 & 0.051 & 0.052 & 0.051 & 0.051 & 0.051 & 0.051 & 0.053 \\ \hline

5,15 &  MSLR & 0.056 & 0.050 & 0.051 & 0.050 & 0.050 & 0.050 & 0.048 & 0.050 & 0.050 & 0.052 \\
&  FZ & 0.045 & 0.046 & 0.051 & 0.046 & 0.047 & 0.044 & 0.042 & 0.044 & 0.046 & 0.044 \\
& GV & 0.047 & 0.044 & 0.047 & 0.045 & 0.044 & 0.047 & 0.048 & 0.045 & 0.047 & 0.045 \\ \hline

5,25 &  MSLR & 0.049 & 0.053 & 0.048 & 0.048 & 0.051 & 0.048 & 0.049 & 0.050 & 0.052 & 0.049 \\
&  FZ & 0.052 & 0.047 & 0.046 & 0.050 & 0.045 & 0.045 & 0.045 & 0.045 & 0.044 & 0.046 \\
& GV & 0.040 & 0.036 & 0.038 & 0.033 & 0.036 & 0.037 & 0.035 & 0.036 & 0.037 & 0.039 \\ \hline

\end{tabular}
\end{center}
\end{table}

\begin{table}
\begin{center}
\caption{ Empirical powers of the tests at nominal level $\alpha =0.05$ with ${\rho}_1=0.05$.} \label{tab.2.1}

\begin{tabular}{|c|l|ccccccccc|} \hline

&  & \multicolumn{9}{|c|}{${\rho }_2$} \\ \hline
$n_1,n_2$ & Method & 0.15 & 0.25 & 0.35 & 0.45 & 0.55 & 0.65 & 0.75 & 0.85 & 0.95 \\ \hline

5,5 &  MSLR & 0.052 & 0.059 & 0.069 & 0.076 & 0.087 & 0.102 & 0.113 & 0.131 & 0.143 \\
& FZ & 0.049 & 0.051 & 0.056 & 0.066 & 0.071 & 0.092 & 0.098 & 0.119 & 0.129 \\
& GV & 0.058 & 0.061 & 0.070 & 0.080 & 0.091 & 0.121 & 0.125 & 0.134 & 0.155 \\ \hline

5,10 &  MSLR & 0.055 & 0.064 & 0.071 & 0.085 & 0.105 & 0.131 & 0.151 & 0.175 & 0.208 \\
& FZ & 0.046 & 0.056 & 0.064 & 0.067 & 0.082 & 0.101 & 0.124 & 0.142 & 0.168 \\
& GV & 0.056 & 0.057 & 0.063 & 0.074 & 0.085 & 0.108 & 0.121 & 0.139 & 0.171 \\ \hline

10,10 &  MSLR & 0.057 & 0.067 & 0.088 & 0.097 & 0.151 & 0.201 & 0.245 & 0.290 & 0.356 \\
& FZ & 0.051 & 0.068 & 0.091 & 0.117 & 0.153 & 0.208 & 0.240 & 0.297 & 0.349 \\
& GV & 0.057 & 0.066 & 0.094 & 0.122 & 0.153 & 0.204 & 0.252 & 0.294 & 0.355 \\ \hline

15,10 &  MSLR & 0.063 & 0.082 & 0.111 & 0.164 & 0.226 & 0.290 & 0.372 & 0.448 & 0.523 \\
& FZ & 0.048 & 0.072 & 0.094 & 0.133 & 0.181 & 0.235 & 0.291 & 0.348 & 0.414 \\
& GV & 0.055 & 0.070 & 0.095 & 0.135 & 0.185 & 0.236 & 0.291 & 0.363 & 0.410 \\ \hline

5,15 & MSLR & 0.061 & 0.063 & 0.088 & 0.099 & 0.115 & 0.136 & 0.192 & 0.214 & 0.253 \\
& FZ & 0.046 & 0.052 & 0.067 & 0.080 & 0.100 & 0.119 & 0.139 & 0.167 & 0.193 \\
& GV & 0.044 & 0.048 & 0.057 & 0.070 & 0.079 & 0.099 & 0.123 & 0.152 & 0.173 \\ \hline

5,25 & MSLR & 0.069 & 0.076 & 0.089 & 0.108 & 0.123 & 0.171 & 0.189 & 0.235 & 0.264 \\
& FZ & 0.052 & 0.059 & 0.067 & 0.084 & 0.102 & 0.122 & 0.140 & 0.162 & 0.196 \\
& GV & 0.043 & 0.042 & 0.054 & 0.064 & 0.073 & 0.091 & 0.114 & 0.133 & 0.158 \\ \hline

20,20 & MSLR & 0.073 & 0.111 & 0.172 & 0.251 & 0.377 & 0.456 & 0.528 & 0.652 & 0.709 \\
& FZ & 0.074 & 0.105 & 0.178 & 0.252 & 0.351 & 0.448 & 0.540 & 0.638 & 0.717 \\
& GV & 0.077 & 0.112 & 0.177 & 0.255 & 0.355 & 0.435 & 0.546 & 0.643 & 0.710 \\ \hline

25,25 & MSLR & 0.078 & 0.123 & 0.201 & 0.324 & 0.442 & 0.549 & 0.638 & 0.776 & 0.823 \\
& FZ & 0.080 & 0.135 & 0.214 & 0.313 & 0.422 & 0.541 & 0.650 & 0.741 & 0.814 \\
& GV & 0.079 & 0.133 & 0.212 & 0.317 & 0.425 & 0.538 & 0.651 & 0.739 & 0.813 \\ \hline

\end{tabular}
\end{center}
\end{table}

\begin{table}
\begin{center}
\caption{ Empirical powers of the tests at nominal level $\alpha =0.05$ with ${\rho}_1=0.05$.}  \label{tab.2.2}

\begin{tabular}{|c|l|ccccccccc|} \hline

&  & \multicolumn{9}{|c|}{${\rho }_2$} \\ \hline
$n_1,n_2$ & Method & -0.15 & -0.25 & -0.35 & -0.45 & -0.55 & -0.65 & -0.75 & -0.85 & -0.95 \\ \hline

5,5 &  MSLR & 0.054 & 0.066 & 0.074 & 0.085 & 0.106 & 0.121 & 0.144 & 0.161 & 0.179 \\
& FZ & 0.049 & 0.058 & 0.060 & 0.075 & 0.090 & 0.105 & 0.120 & 0.142 & 0.158 \\
& GV & 0.063 & 0.071 & 0.080 & 0.090 & 0.109 & 0.119 & 0.143 & 0.166 & 0.192 \\ \hline

5,10 &  MSLR & 0.058 & 0.072 & 0.089 & 0.114 & 0.133 & 0.166 & 0.194 & 0.216 & 0.253 \\
& FZ & 0.051 & 0.068 & 0.074 & 0.092 & 0.108 & 0.130 & 0.154 & 0.174 & 0.202 \\
& GV & 0.056 & 0.061 & 0.074 & 0.092 & 0.107 & 0.126 & 0.150 & 0.181 & 0.199 \\ \hline

10,10 &  MSLR & 0.066 & 0.084 & 0.116 & 0.164 & 0.205 & 0.252 & 0.308 & 0.368 & 0.430 \\
& FZ & 0.068 & 0.090 & 0.117 & 0.162 & 0.206 & 0.248 & 0.304 & 0.364 & 0.425 \\
& GV & 0.071 & 0.089 & 0.120 & 0.164 & 0.208 & 0.257 & 0.312 & 0.368 & 0.424 \\ \hline

15,10 &  MSLR & 0.072 & 0.113 & 0.121 & 0.192 & 0.247 & 0.278 & 0.385 & 0.445 & 0.505 \\
& FZ & 0.069 & 0.099 & 0.121 & 0.188 & 0.247 & 0.304 & 0.372 & 0.440 & 0.511 \\
& GV & 0.071 & 0.100 & 0.141 & 0.186 & 0.252 & 0.299 & 0.368 & 0.442 & 0.512 \\ \hline

5,15 & MSLR & 0.058 & 0.064 & 0.080 & 0.101 & 0.116 & 0.136 & 0.179 & 0.188 & 0.254 \\
& FZ & 0.053 & 0.058 & 0.068 & 0.078 & 0.101 & 0.122 & 0.143 & 0.166 & 0.193 \\
& GV & 0.048 & 0.055 & 0.061 & 0.077 & 0.086 & 0.100 & 0.127 & 0.147 & 0.176 \\ \hline

5,25 & MSLR & 0.061 & 0.071 & 0.083 & 0.092 & 0.118 & 0.158 & 0.196 & 0.240 & 0.268 \\
& FZ & 0.053 & 0.057 & 0.068 & 0.078 & 0.101 & 0.122 & 0.143 & 0.166 & 0.193 \\
& GV & 0.038 & 0.044 & 0.055 & 0.065 & 0.081 & 0.098 & 0.113 & 0.135 & 0.164 \\ \hline

20,20 & MSLR & 0.063 & 0.103 & 0.183 & 0.246 & 0.366 & 0.417 & 0.565 & 0.660 & 0.709 \\
& FZ & 0.075 & 0.119 & 0.179 & 0.249 & 0.349 & 0.439 & 0.546 & 0.630 & 0.716 \\
& GV & 0.071 & 0.114 & 0.176 & 0.257 & 0.346 & 0.448 & 0.552 & 0.638 & 0.717 \\ \hline

25,25 & MSLR & 0.064 & 0.145 & 0.207 & 0.319 & 0.418 & 0.569 & 0.620 & 0.734 & 0.819 \\
& FZ & 0.077 & 0.137 & 0.211 & 0.319 & 0.426 & 0.541 & 0.651 & 0.750 & 0.813 \\
& GV & 0.079 & 0.133 & 0.212 & 0.314 & 0.432 & 0.542 & 0.651 & 0.735 & 0.821 \\ \hline

\end{tabular}
\end{center}
\end{table}

\newpage

\subsection{Real Data}

In this example, we test the equality of two independent correlations in three groups of data. This data set is given by \cite{bi-br-gu-04}
and also  have been analyzed by \cite{kr-xi-07}. 

For each of two groups of 14 men and 14 women, the  sample correlation between a verbal memory score (v) and laterality of blood flow in each of three brain regions, namely, temporal (t), frontal (f) and subcortical (s) are obtained in Table \ref{tab.3}. It may be of interest to compare the correlations between the laterality of blood flow in each brain region and verbal memory score across gender. The results are given in Table \ref{tab.4}. We can find that there is no significant difference between male and female correlations in frontal and subcortical cases.

\begin{table}[ht]
\begin{center}
\caption{ correlations between laterality of blood flow in three brain regions and verbal memory score.}\label{tab.3}
\begin{tabular}{|c|c|c|c|}
\hline
Gender & \multicolumn{3}{|c|}{laterality of blood flow} \\ \hline
& Temporal & Subcortical & Frontal \\ \hline
Male & $r_{\mathrm{M,vt}}=-0.340$ & $r_{\mathrm{M,vs}}=0.641$ & $r_{\mathrm{%
M,vf}}=-0.032$ \\ \hline
Female & $r_{\mathrm{F,vt}}=0.812$ & $r_{\mathrm{F,vs}}=0.491$ & $r_{\mathrm{%
F,vf}}=-0.212$ \\ \hline
\end{tabular}
\end{center}
\end{table}

\begin{table}[ht]
\begin{center}
\caption{ P-values of the tests for equality of correlations between laterality of blood flow in three brain regions and verbal memory score.}\label{tab.4}

\begin{tabular}{|l|c|c|c|}
\hline
method & ${\rho }_{\mathrm{M,vt}}={\rho }_{\mathrm{F,vt}}$ & ${\rho }_{%
\mathrm{M,vs}}={\rho }_{\mathrm{F,vs}}$ & ${\rho }_{\mathrm{M,vf}}={\rho }_{%
\mathrm{F,vf}}$ \\ \hline
MSLR & 0.0008 & 0.5978 & 0.6677 \\ \hline
GV & 0.0008 & 0.5948 & 0.6682 \\ \hline
Fisher Z & 0.0004 & 0.6018 & 0.6673 \\ \hline
\end{tabular}
\end{center}
\end{table}

\section{Conclusion}
\label{sec.con}
%
%

{
Existing methods for comparing the correlation coefficients of two independent bivariate normal distributions do not perform well in a range of
small-sample settings.
\cite{kr-xi-07}
obtained a generalized pivotal quantity for difference of two correlation coefficients and they gave a method for testing the equality of two correlation coefficients using this generalized pivotal quantity. By using a simulation study, they showed the test size of their method is greater than the nominal level for small sample sizes. In other words, they showed that their method is liberal.

Using the method of modified signed log likelihood,
\cite{kr-le-14}
considered the problem of equality of coefficients of variation for independent normal populations. This method is an exact method to test a hypothesis for unknown parameter. The accuracy of this method is very satisfactory such that the actual size of test is approximately close to nominal level even for small sample sizes.
In this article, we used the MSLR method for testing the equality of two independent correlation coefficients because of the accuracy of this method for inference about the unknown model parameter and that the other competing methods have deviations that we cannot rely to them.
In this paper, we explained the generating the sample correlation coefficients of a bivariate normal distribution (See Lemma 2.1). Then, we obtained the MSLR test statistics for the problem of the equality of two correlation coefficients.
All the mathematical formulas obtained in our paper are different from that of used in chapter 4 of
\cite{kr-xi-07}.
 As we see, for using the MSLR method, one should consider the MLEs of the unknown parameters but the MLE was not stated in
 \cite{kr-xi-07}
at all.
We compared our method with the method of GV used by
\cite{kr-xi-07}
 and another competing method by simulation studies. We found that MSLR test always are satisfactory in terms of actual sizes regardless of the sample sizes and values of the common correlation coefficient in comparison with existing methods. Also, the power of MSLR is larger than powers of GV and Fisher Z-transform when the sample sizes are unequal.
}

\bibliographystyle{apa}

\begin{thebibliography}{}

\bibitem[\protect\astroncite{Barndorff-Nielsen}{1986}]{barndorff-86}
Barndorff-Nielsen, O.~E. (1986).
\newblock Inference on full or partial parameters based on the standardized
  signed log likelihood ratio.
\newblock {\em Biometrika}, 73(2):307--322.

\bibitem[\protect\astroncite{Barndorff-Nielsen}{1991}]{barndorff-91}
Barndorff-Nielsen, O.~E. (1991).
\newblock Modified signed log likelihood ratio.
\newblock {\em Biometrika}, 78(3):557--563.

\bibitem[\protect\astroncite{Bilker et~al.}{2004}]{bi-br-gu-04}
Bilker, W.~B., Brensinger, C., and Gur, R.~C. (2004).
\newblock A two factor {ANOVA}-like test for correlated correlations:
  {CORANOVA}.
\newblock {\em Multivariate Behavioral Research}, 39(4):565--594.

\bibitem[\protect\astroncite{Cox and Hinkley}{1979}]{co-hi-74}
Cox, D.~R. and Hinkley, D.~V. (1979).
\newblock {\em Theoretical Statistics}.
\newblock Chapman and Hall, London.

\bibitem[\protect\astroncite{DiCiccio et~al.}{2001}]{di-ma-st-01}
DiCiccio, T.~J., Martin, M.~A., and Stern, S.~E. (2001).
\newblock Simple and accurate one-sided inference from signed roots of
  likelihood ratios.
\newblock {\em Canadian Journal of Statistics}, 29(1):67--76.

\bibitem[\protect\astroncite{Donner and Rosner}{1980}]{do-ro-80}
Donner, A. and Rosner, B. (1980).
\newblock On inferences concerning a common correlation coefficient.
\newblock {\em Applied Statistics}, 29:69--76.

\bibitem[\protect\astroncite{Fisher}{1915}]{fisher-15}
Fisher, R.~A. (1915).
\newblock Frequency distribution of the values of the correlation coefficient
  in samples from an indefinitely large population.
\newblock {\em Biometrika}, 10(4):507--521.

\bibitem[\protect\astroncite{Fisher}{1921}]{fisher-21}
Fisher, R.~A. (1921).
\newblock On the ``probable error'' of a coefficient of correlation deduced
  from a small sample.
\newblock {\em Metron}, 1:3--32.

\bibitem[\protect\astroncite{Hotelling}{1953}]{hotelling-53}
Hotelling, H. (1953).
\newblock New light on the correlation coefficient and its transforms.
\newblock {\em Journal of the Royal Statistical Society. Series B
  (Methodological)}, 15(2):193--232.

\bibitem[\protect\astroncite{Kazemi and Jafari}{2015}]{ka-ja-15}
Kazemi, M.~R. and Jafari, A.~A. (2015).
\newblock Comparing seventeen interval estimates for a bivariate normal
  correlation coefficient.
\newblock {\em Journal of Statistics Applications \& Probability Letters},
  2(1):23--35.

\bibitem[\protect\astroncite{Krishnamoorthy and Lee}{2014}]{kr-le-14}
Krishnamoorthy, K. and Lee, M. (2014).
\newblock Improved tests for the equality of normal coefficients of variation.
\newblock {\em Computational Statistics}, 29(1-2):215--232.

\bibitem[\protect\astroncite{Krishnamoorthy and Xia}{2007}]{kr-xi-07}
Krishnamoorthy, K. and Xia, Y. (2007).
\newblock Inferences on correlation coefficients: One-sample, independent and
  correlated cases.
\newblock {\em Journal of Statistical Planning and Inference},
  137(7):2362--2379.

\bibitem[\protect\astroncite{Muirhead}{1982}]{muirhead-82}
Muirhead, R.~J. (1982).
\newblock {\em Aspects of Multivariate Statistical Theory}.
\newblock Wiley, New York.

\bibitem[\protect\astroncite{Olkin and Finn}{1995}]{ol-fi-95}
Olkin, I. and Finn, J.~D. (1995).
\newblock Correlations redux.
\newblock {\em Psychological Bulletin}, 118:155--164.

\bibitem[\protect\astroncite{Pearson}{1933}]{pearson-33}
Pearson, K. (1933).
\newblock On a method of determining whether a sample of size n supposed to
  have been drawn from a parent population having a known probability integral
  has probably been drawn at random.
\newblock {\em Biometrika}, 25:379--410.

\bibitem[\protect\astroncite{Pierce and Peters}{1992}]{pi-pe-92}
Pierce, D.~A. and Peters, D. (1992).
\newblock Practical use of higher order asymptotics for multiparameter
  exponential families.
\newblock {\em Journal of the Royal Statistical Society. Series B
  (Methodological)}, pages 701--737.

\bibitem[\protect\astroncite{Skovgaard}{2001}]{Skovgaard-01}
Skovgaard, I.~M. (2001).
\newblock Likelihood asymptotics.
\newblock {\em Scandinavian Journal of Statistics}, 28(1):3--32.

\bibitem[\protect\astroncite{Sun and Wong}{2007}]{su-wo-07}
Sun, Y. and Wong, A. (2007).
\newblock Interval estimation for the normal correlation coefficient.
\newblock {\em Statistics and Probability Letters}, 77(17):1652--1661.

\bibitem[\protect\astroncite{Zar}{1999}]{zar-99}
Zar, J.~H. (1999).
\newblock {\em Biostatistical analysis}.
\newblock Prentice Hall, India, 4th edition.

\end{thebibliography}

\end{document}